\documentclass[submission,copyright,creativecommons]{eptcs}
\providecommand{\event}{QPL 2019} 
\usepackage{breakurl}             
\usepackage{underscore}           

\usepackage[utf8]{inputenc}
\usepackage[T1]{fontenc}

\usepackage{algorithm2e}
\usepackage{amsmath}
\usepackage{amssymb}
\usepackage{color}
\usepackage[table]{xcolor}
\usepackage{amsthm}
\usepackage{hyperref}
\usepackage[caption=no]{subfig}
\usepackage{tikz}
\usetikzlibrary{calc,positioning}

\newtheorem{example}{Example}
\newtheorem{lemma}{Lemma}

\newtheorem{corollary}{Corollary}

\newtheorem{definition}{Definition}

\newcommand{\nonsupersets}{\ensuremath{\mathbin{\text{\raisebox{1pt}{$\scriptstyle\searrow$}}}}}

\newcommand{\HI}{\ensuremath{\text{\fontsize{9pt}{9pt}\selectfont HI}}}
\newcommand{\LO}{\ensuremath{\text{\fontsize{9pt}{9pt}\selectfont LO}}}
\newcommand{\CHOOSE}{\ensuremath{\text{\fontsize{9pt}{9pt}\selectfont CHOOSE}}}
\newcommand{\ZUNIQUE}{\ensuremath{\text{\fontsize{9pt}{9pt}\selectfont ZUNIQUE}}}

\newcommand{\FROM}{\ensuremath{\textit{from}}}
\newcommand{\TO}{\ensuremath{\textit{to}}}
\newcommand{\VALID}{\ensuremath{\textit{valid}}}
\newcommand{\BAD}{\ensuremath{\textit{bad}}}
\newcommand{\MAP}{\ensuremath{\textit{map}}}
\newcommand{\EDGES}{\ensuremath{\textit{edges}}}
\newcommand{\LAYERS}{\ensuremath{\textit{layers}}}

\definecolor{mygrey}{gray}{0.85}

\title{Using ZDDs in the Mapping of Quantum Circuits}
\author{Kaitlin Smith  \qquad Mitchell Thornton
\institute{Quantum Informatics Research Group, SMU\\
Dallas, Texas, USA}
\and
Mathias Soeken \qquad Bruno Schmitt\qquad Giovanni De Micheli
\institute{Integrated Systems Laboratory, EPFL\\
Lausanne, Switzerland}
}
\def\titlerunning{Using ZDDs in the Mapping of Quantum Circuits}
\def\authorrunning{Smith, Soeken, Schmitt, De Micheli, and Thornton}
\begin{document}
\maketitle

\begin{abstract}
A critical step in quantum compilation is the transformation of a technology-independent quantum circuit into a technology-dependent form for a targeted device. In addition to mapping quantum gates into the supported gate set, it is necessary to map pseudo qubits in the technology-independent circuit into physical qubits of the technology-dependent circuit such that coupling constraints among qubits acting in multiple-qubit gates are satisfied. It is usually not possible to find such a mapping without adding SWAP gates into the circuit. To cope with the technical limitations of NISQ-era quantum devices, it is advantageous to find a mapping that requires as few additional gates as possible. The large search space of possible mappings makes this task a difficult combinatorial optimization problem. In this work, we demonstrate how zero-suppressed decision diagrams (ZDDs) can be used for typical implementation tasks in quantum mapping algorithms. We show how to maximally partition a quantum circuit into blocks of adjacent gates, and if adjacent gates within a circuit do not share common mapping permutations, we attempt to combine them using parallelized SWAP operations represented in a ZDD. Boundaries for the partitions are formed where adjacent gates are unable to be combined. Within each partition block, ZDDs represent all possible mappings of pseudo qubits to physical qubits. 
\end{abstract}

\section{Introduction}
\label{Introduction}
Today’s NISQ-era quantum devices~\cite{Preskill18} support some given set of single- and two-qubit quantum gates. While single-qubit operations can be executed on any of the physical qubits, two-qubit quantum gates can only be performed by a pair of qubits that share a physical connection. The set of permissible qubit pairs are referred to as the coupling constraints. One task in quantum compilation algorithms is the mapping of a quantum circuit or algorithm, a sequence of quantum operations, onto the physical qubits of the device such that all two-qubit operations are executed with respect to device coupling constraints. This task is not always possible without including additional gates in the circuit.

Finding an optimum solution that minimizes the number of additional gates is ``NP-hard''~\cite{BKM18}. In order to efficiently find a solution, several heuristics have been proposed~\cite{SSCP18,ZPW18,HY18,LDX18}. Past work in mapping algorithms for quantum circuits are reported in~\cite{childs2019circuit,kate-ulsi-ws}. A common bottleneck in these heuristic methods is due to the large combinational search space resulting in numerous possible ways of mapping the gates to the device.

In this paper, we discuss how zero-suppressed decision diagrams (ZDDs,~\cite{Minato93,Knuth4A}) can be used in mapping algorithms to combat combinational complexity. These data structures were selected because nearest neighbor couplings within quantum devices make the possible connections between qubits in a quantum circuit a sparse set, and ZDDs, as compared to other types of decision diagrams, are efficient at representing sparse sets. We show how to implement two specific problems which appear in several heuristics: (1) finding a maximal subcircuit partition that can be mapped without adding gates, and (2) how to determine and choose among all possible SWAP circuits that can execute in parallel in order to extend a partition. Once a maximal subcircuit partition is determined, a pseudo to physical qubit permutation from that partition is used to map the quantum circuit to a technology platform. If the maximal partition spans the entire quantum circuit, the result of ZDD mapping is a circuit that is fully mapped with respect to the coupling constraints of a device.

Finding maximal partitions for quantum circuit mapping is solved using SAT in~\cite{HY18}. In contrast to the SAT-based solution that finds one possible mapping, the ZDD-based algorithm generates all possible mappings for a maximal partition. All solutions are represented implicitly by means of a decision diagram, which may be used to count all solutions, query some solutions, or compute the solution that minimizes some linear cost function. All such tasks can be performed in time that is linear with respect to the size of the ZDD. All possible parallel SWAP operations used to extend the size of a partition are also stored in a ZDD.

The algorithms reported in this paper are intended to serve as motivating examples to illustrate how ZDDs may be used as a data structure for implementing a mapping algorithm. These algorithms have been prototyped and evaluated on a set of benchmarks. Effectiveness of the ZDD mapping algorithms is determined by using them in a preprocessing step for the publically-available compilers developed for the IBM and Rigetti quantum devices.

\section{Preliminaries}
\label{Preliminaries}
\subsection{Graphs}
\label{Graphs}
An undirected graph $G = (V, E)$ consists of a set of vertices $V$ and a set of
edges $E \subseteq \binom{V}{2} = \{V' \subseteq V\mid |V'|=2\}$. Given two
undirected graphs $G_1 = (V_1, E_1)$ and $G_2 = (V_2, E_2)$, we say that $G_1$
is a subgraph of $G_2$, if there exists an injective function $\varphi : V_1 \to
V_2$ such that $\{v, w\} \in E_1$ implies $\{f(v), f(w)\} \in E_2$. A directed graph, $G = (V, E)$, is similar to a undirected graph, but edges consist of ordered pairs, or permutations, of nodes rather than bidirectional combinations. Therefore, in a directed graph, $E \subseteq P^{V}_{2} = \{V' \subseteq V\mid |V'|=2\}$ where $P^{n}_{k}$ indicates a $k$-permutation of $n$.

Given an ordered sequence $S = s_1, \dots, s_n$, we define an ordered partition
$S_1, \dots, S_l$ as a set of nonempty subsequences $S_i = s_{b_i}, s_{b_i+1},
\dots, s_{e_i}$ such that $s_{b_1} = s_1$, $s_{e_l} = s_n$, and $s_{b_i} = s_{e_{i-1} + 1}$ for
all $1 < i \le l$.

\begin{example}
  Let $S = 3, 5, 1, 2, 8, 2, 3, 4$.  Then $S_1 = 3, 5, 1$, $S_2 = 2, 8$, and
  $S_3 = 2, 3, 4$ is an ordered partition of $S$.
\end{example}

\subsection{Zero-suppressed decision diagrams}

Given a set of variables $X = \{x_1, \dots, x_n\}$, a
ZDD~\cite{Minato93,Knuth4A} is a directed acyclic graph with
nonterminal vertices $N$ and two terminal vertices $\top$ and $\bot$.
Each non-terminal vertex $v \in N$ is associated with a variable
$V(v) \in \{1, \dots, n\}$ and two successor nodes
$\HI(v), \LO(v) \in N \cup \{\top, \bot\}$.  The nodes on a path follow a fixed variable order on the way to a terminal node.  We have
$\HI(v) \in \{\top, \bot\}$ or $V(\HI(v)) > V(v)$ for all $v$.\footnote{To simplify the presentation in the paper, we assume the
  variable ordering $1 < 2 < \cdots < n$.  In practice, any
  permutation of this order can be used.}  The same applies to
$\LO(v)$.

Each vertex in the ZDD represents a finite family of finite subsets
over $X$ where families of sets are canonical up to order of the sets and repetition. The terminal node $\bot$ represents the \emph{empty family}
$\emptyset$ and the terminal node $\top$ represents the \emph{unit
  family} which is the set containing the empty set $\{\emptyset\}$.
Each non-terminal $v$ represents the subset
\begin{equation}
  \label{eq:zdd-recursive}
  \LO(v) \cup \{S \cup \{x_{V(v)}\} \mid S \in \HI(v)\}.
\end{equation}

A ZDD is \emph{reduced} if there are no two vertices that represent
the same sets.  This implies that in a reduced ZDD there cannot be a
vertex $v$ with $\HI(v) = \bot$, since such a vertex represents the
set $\LO(v)$.  For the sake of convenience, we use $\epsilon_x$ to
denote the \emph{elementary family} $\{\{x\}\}$ for each $x \in X$.
Finally, we use $\wp$ to refer to the ZDD that represents the
\emph{universal family} of all subsets of $X$.

\begin{figure}[t]
 \centering

  \begin{tikzpicture}[font=\small,x=.55cm,y=.55cm]
    \tikzstyle{term}=[draw,inner sep=1pt];
    \tikzstyle{nonterm}=[draw,circle,inner sep=1pt];
    \tikzstyle{lo}=[dash pattern=on 2pt off 1pt];
    \node[term] at (-.5,0) (bot) {$\bot$};
    \node[term] at (.5,0) (top) {$\top$};
    \node[nonterm] at (0,1) (4) {$4$};
    \node[nonterm] at (-.5,2) (3a) {$3$};
    \node[nonterm] at (.5,2) (3b) {$3$};
    \node[nonterm] at (-.5,3) (2a) {$2$};
    \node[nonterm] at (.5,3) (2b) {$2$};
    \node[nonterm] at (0,4) (1) {$1$};
    \draw[lo] (1) -- (2a);
    \draw (1) -- (2b);
    \draw[lo] (2a) -- (3a);
    \draw (2a) -- (3b);
    \draw[lo] (2b) -- (3b);
    \draw (2b) to[out=-60,in=60] (top);
    \draw[lo] (3a) -- (bot);
    \draw (3a) -- (4);
    \draw[lo] (3b) -- (4);
    \draw (3b) -- (top);
    \draw[lo] (4) -- (bot);
    \draw (4) -- (top);

    \begin{scope}[font=\footnotesize]
      \node[right=.3cm of 2b,anchor=west] (s-2b) {$\{\{x_2\}, \{x_3\}, \{x_4\}\}$};
      \node[left=.3cm of 2a,anchor=east] (s-2a) {$\{\{x_2, x_3\}, \{x_2, x_4\}, \{x_3, x_4\}\}$};
      \node[right=.3cm of 3b,anchor=west] (s-3b) {$\{\{x_3\}, \{x_4\}\}$};
      \node[left=.3cm of 3a,anchor=east] (s-3a) {$\{\{x_3, x_4\}\}$};

      \node[left=.3cm of 4,anchor=east] (s-4) {$\{\{x_4\}\}$};
    \end{scope}
  \end{tikzpicture}

  \centering
  \caption{A ZDD representing the family of sets
    $\{\{x_1, x_2\}, \{x_1, x_3\}, \{x_1, x_4\}, \{x_2, x_3\}, \{x_2,
    x_4\}, \{x_3, x_4\}\}$.  All internal non-terminal nodes are
    annotated with the sets they represent. Dashed edges indicate LO 		    and solid edges indicate HI.}
  \label{fig:zdd-example}
\end{figure}
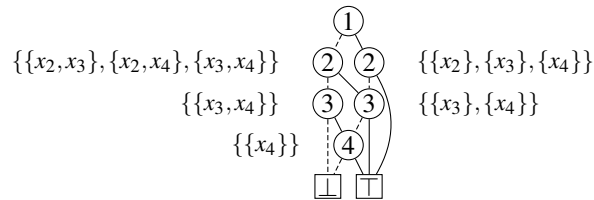

We write $|f|$ to denote the number of sets in a family~$f$.  We write
$Z(f)$ to denote the number of nodes, including the terminal nodes, of
the reduced ZDD for~$f$. It should be noted that because of the reduction rules associated with ZDDs, the data structure is a canonical representation of a function with respect to a fixed variable order.

\begin{example}
  Fig.~\ref{fig:zdd-example} shows a ZDD for
  $f =\{\{x_1, x_2\}, \{x_1, x_3\}, \{x_1, x_4\}, \{x_2, x_3\}, \{x_2,
  x_4\}, \{x_3, x_4\}\}$, i.e., all two-element subsets of
  $X = \{x_1, x_2, x_3, x_4\}$.  We have $|f| = 6$ and $Z(f) = 8$.
  In the general case, the ZDD $f$ that represents all $k$-element
  subsets of a set $\{x_1, \dots, x_n\}$ has $Z(f) = O(kn)$ nodes,
  while representing $|f| = \binom{n}{k}$ sets.
\end{example}

Given two ZDDs $f$ and $g$, the following list of operations is part
of what is called a ZDD family algebra.  Each operation can be
efficiently implemented using ZDDs.
\begin{align*}
  f \cup g &= \{\alpha \mid \text{$\alpha \in f$ or $\alpha \in g$}\} & \textit{union} \\
  f \cap g &= \{\alpha \mid \text{$\alpha \in f$ and $\alpha \in g$}\} & \textit{intersection} \\
  f \setminus g &= \{\alpha \mid \text{$\alpha \in f$ and $\alpha \notin g$}\} & \textit{difference} \\
  f \sqcup g &= \{\alpha \cup \beta \mid \text{$\alpha \in f$ and $\beta \in g$}\} & \textit{join} \\
  f \sqcap g &= \{\alpha \cap \beta \mid \text{$\alpha \in f$ and $\beta \in g$}\} & \textit{meet} \\
  f \nonsupersets g &= \{\alpha \in f \mid \text{$\beta \in g$ implies $\alpha \not\supseteq \beta$}\} & \textit{nonsupersets}
\end{align*}

Finally, if $f$ represents the family
$\epsilon_{x'_1} \cup \cdots \cup \epsilon_{x'_l}$ for some subset
$\{x'_1, \dots, x'_l\} = X' \subseteq X$, then
\[
  \binom{f}{k}
\]
is the ZDD that represents the family $\binom{X'}{k}$.

Note that the \emph{nonsupersets} operation can be described in terms
of the others: $f \nonsupersets g = f \setminus (f \sqcup g)$.
However, it may be more efficient to implement the operation
explicitly in a ZDD package.  For a detailed description of how ZDDs
are represented in memory and how the ZDD family algebra operations
are implemented, the reader is referred to the
literature~\cite{Knuth4A}.

\subsection{Physical Quantum Architectures}
\label{Physical Quantum Architectures}

\subsubsection{Rigetti}
\label{Rigetti QPUs}

Rigetti has developed quantum machines based on solid-state, superconducting circuit technology. The company has also developed the Python library pyQuil as well as a software development kit (SDK) called \textit{Forest} that can be used to write quantum algorithms, to interact with quantum processing units (QPUs), and to simulate quantum computing. The quantum instruction language Quil is used to specify algorithms for the Rigetti QPUs~\cite{smith2016practical}. Within the SDK, it is possible to create custom architectures that can be targeted by the Rigetti compiler. To implement an algorithm on the Rigetti QPUs, complex gates must be decomposed into the native gate set of $R_z(\theta), R_x(\frac{k\pi}{2})$ where $k$ is an integer value, and the two-qubit operator $CZ$. As an additional constraint, $CZ$ operations are limited with respect to where they may be placed on the device due to the constraints among qubits. The Rigetti compiler may be used to transform technology-independent circuits into forms that have the appropriate gate library and connections for QPU execution. This compiler performs mapping and minimization procedures, but the resulting circuits may not necessarily be the optimum solution.

\subsubsection{IBM}
\label{IBM}

IBM has also developed quantum computers based on solid-state, superconducting circuit technology. The Python SDK \textit{Qiskit} is their tool for performing quantum information processing (QIP) with their platform. Quantum assembly language, or QASM, is used to specify quantum logic for execution on the IBM quantum machines~\cite{cross2017open}. In order for a quantum circuit to be executable on a real machine, the QASM specification must not only obey coupling map constrictions but also contain operators within the gate set of $R_z(\theta)$, $R_x(\phi)$, $R_y(\gamma)$, and the two-qubit operator $CX$. The compiler contained in the \textit{Qiskit} SDK may be used for transforming circuits into a technology-ready form, but as with the Rigetti compiler, solutions may not be optimum. \textit{Qiskit} allows for compilation to custom architectures so that devices outside of the existing IBM machines may be targeted.

\section{Problem formulation}
In this paper, we model the quantum circuit that is to be mapped to a
quantum device as a set of pseudo qubits $V = \{v_1, \dots, v_n\}$ and
an ordered sequence of two-qubit gates $G_j = g_1, \dots, g_n$, with
$g_i \in \binom{V}{2}$.  We can safely ignore the one-qubit gates in
the circuit, since the coupling constraints of the device do not
affect their mapping.  Also note that we do not take the direction of
a gate (e.g., the position of control and target in a $CX$) into
account as unidirectional gates may be reversed, as seen in~\cite{SSCP18}, with single qubit operations.

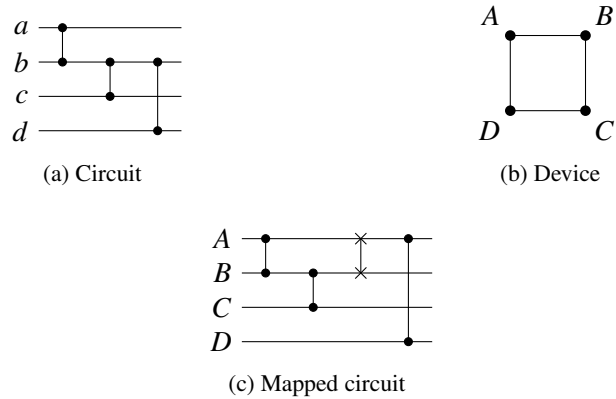
\begin{figure}[t]
  \centering
  \subfloat[Circuit]{\tikzpicture[scale=1.000000,x=1pt,y=1pt]
\filldraw[color=white] (0.000000, -6.500000) rectangle (54.000000, 45.500000);
\draw[color=black] (0.000000,39.000000) -- (54.000000,39.000000);
\draw[color=black] (0.000000,39.000000) node[left] {$a$};
\draw[color=black] (0.000000,26.000000) -- (54.000000,26.000000);
\draw[color=black] (0.000000,26.000000) node[left] {$b$};
\draw[color=black] (0.000000,13.000000) -- (54.000000,13.000000);
\draw[color=black] (0.000000,13.000000) node[left] {$c$};
\draw[color=black] (0.000000,0.000000) -- (54.000000,0.000000);
\draw[color=black] (0.000000,0.000000) node[left] {$d$};
\draw (9.000000,39.000000) -- (9.000000,26.000000);
\filldraw (9.000000, 39.000000) circle(1.500000pt);
\filldraw (9.000000, 26.000000) circle(1.500000pt);
\draw (27.000000,26.000000) -- (27.000000,13.000000);
\filldraw (27.000000, 26.000000) circle(1.500000pt);
\filldraw (27.000000, 13.000000) circle(1.500000pt);
\draw (45.000000,26.000000) -- (45.000000,0.000000);
\filldraw (45.000000, 26.000000) circle(1.500000pt);
\filldraw (45.000000, 0.000000) circle(1.500000pt);
\endtikzpicture}\hfil
  \subfloat[Device]{
\begin{tikzpicture}
  \fill (0,1) coordinate (A) circle (2pt) node[above left] {$A$};
  \fill (1,1) coordinate (B) circle (2pt) node[above right] {$B$};
  \fill (1,0) coordinate (C) circle (2pt) node[below right] {$C$};
  \fill (0,0) coordinate (D) circle (2pt) node[below left] {$D$};

  \draw (A) -- (B) -- (C) -- (D) -- (A);
\end{tikzpicture}}

  \subfloat[Mapped circuit]{\tikzpicture[scale=1.000000,x=1pt,y=1pt]
\filldraw[color=white] (0.000000, -6.500000) rectangle (72.000000, 45.500000);
\draw[color=black] (0.000000,39.000000) -- (72.000000,39.000000);
\draw[color=black] (0.000000,39.000000) node[left] {$A$};
\draw[color=black] (0.000000,26.000000) -- (72.000000,26.000000);
\draw[color=black] (0.000000,26.000000) node[left] {$B$};
\draw[color=black] (0.000000,13.000000) -- (72.000000,13.000000);
\draw[color=black] (0.000000,13.000000) node[left] {$C$};
\draw[color=black] (0.000000,0.000000) -- (72.000000,0.000000);
\draw[color=black] (0.000000,0.000000) node[left] {$D$};
\draw (9.000000,39.000000) -- (9.000000,26.000000);
\filldraw (9.000000, 39.000000) circle(1.500000pt);
\filldraw (9.000000, 26.000000) circle(1.500000pt);
\draw (27.000000,26.000000) -- (27.000000,13.000000);
\filldraw (27.000000, 26.000000) circle(1.500000pt);
\filldraw (27.000000, 13.000000) circle(1.500000pt);
\draw (45.000000,39.000000) -- (45.000000,26.000000);
\scope
\draw (42.878680, 36.878680) -- (47.121320, 41.121320);
\draw (42.878680, 41.121320) -- (47.121320, 36.878680);
\endscope
\scope
\draw (42.878680, 23.878680) -- (47.121320, 28.121320);
\draw (42.878680, 28.121320) -- (47.121320, 23.878680);
\endscope
\draw (63.000000,39.000000) -- (63.000000,0.000000);
\filldraw (63.000000, 39.000000) circle(1.500000pt);
\filldraw (63.000000, 0.000000) circle(1.500000pt);
\endtikzpicture}
  \caption{Simple circuit and device}
  \label{fig:circarch}
\end{figure}

\begin{example}
  Fig.~\ref{fig:circarch}(a) shows a quantum circuit on four pseudo
  qubits $V = \{a, b, c, d\}$ and three two-qubit gates
  $g_1 = \{a, b\}$, $g_2 = \{b, c\}$, and $g_3 = \{b, d\}$.
\end{example}

A quantum device is modeled by an undirected graph $(P,E)$, where $P = \{p_1,
\dots, p_m\}$ is a set of physical qubits and an edge $\{p,q\} \in E \subseteq
\binom{P}{2}$ states that a 2-qubit operation can be executed using the two
physical qubits $p$ and $q$.

\begin{example}
  Fig.~\ref{fig:circarch}(b) shows a simple four-qubit quantum device with physical
  qubits $P = \{A, B, C, D\}$ and a ring topology, i.e., edges $E = \{\{A, B\},
  \{B, C\}, \{C, D\}, \{D, A\}\}.$
\end{example}

The goal is to find a mapping $\varphi : V \to P$ of pseudo qubits
into physical ones, such that all two-qubit operations in the circuit
are executed on adjacent qubits according to the device's coupling
constraints.  It may not be possible to find such mapping for the
input circuit.  For example, there exists no such mapping for the
example circuit, since the pseudo qubit $b$ interacts with three other
qubits.  However, by adding SWAP gates to reorder pseudo qubits, a
mapping can be achieved. A SWAP gate is a two-qubit operation that can either be implemented with $CX$ or single-qubit rotations and $CZ$.

\begin{example}
  By adding a SWAP gate after the second gate of the circuit in
  Fig.~\ref{fig:circarch}(a), one must update all successive gates in order to
  retain the functionality as is seen in Fig.~\ref{fig:circarch}(c).  However, the transformed circuit can now be mapped to
  the quantum device by mapping $a \mapsto A$, $b \mapsto B$, $c \mapsto C$, and
  $d \mapsto D$.
\end{example}

The aim is to use a small number of SWAP gates when transforming an initial
circuit to a circuit that can be mapped into a target device.  Finding the
globally optimum mapping and a transformed circuit using the fewest number of
SWAP gates is a computationally complex and time-consuming task~\cite{BKM18}.

To address the problem of finding maximal partitions for a circuit using ZDDs, we use partitions where

 \begin{itemize}
\item there exists a $(V, G_j)$ that is a subgraph of $(P, E)$ for $1 \le j \le l$,
\item there is not a $(V, G_j \cup \{g_{{e_j+1}}\})$ that is a subgraph of $(P, E)$ for $1 \le j < l$.
\end{itemize}

Each partition has an associated set of mappings of pseudo to physical qubits $\Phi_j = \{\varphi : V \mapsto P \}\ $ where $\varphi$ is a subgraph isomorphism of $(V, G_j)$ to $(P, E)$. If partition $G_j$ starting at gate $g_b$ cannot be extended, SWAP operations are inserted to merge the last gate of the partition, $g_e$ with the adjacent gate, $g_{e+1}$, in the circuit. These swapping operations, referred to as $\LAYERS$, exchange information on the adjacent physical qubits of the device and are executable in parallel within a single time cycle. The best SWAP layer is chosen according to a scoring metric. Once selected, the SWAP layer merges the gate $g_{i}$ with $g_{i+1}$ by inserting SWAP gates before $g_{i+1}$, extending $G_j$.

Ideally, a single partition will cover the entire circuit, providing a set of mappings that assign pseudo qubits to physical qubits on the device. In the case that multiple partitions exist that are fully extended with inserted SWAP $\LAYERS$, a mapping for the circuit is selected using the largest, and therefore maximal, partition.

\begin{definition} \textit{Maximal Partition} \label{def:maxpart} \\
A maximal partition within a quantum circuit is the largest subset of gates $G_j$ that can be covered by a set of mappings of pseudo to physical qubits, $\Phi_j $. Maximal partitions are fully expanded with parallel SWAP operations that exchange information between physical qubits.
\end{definition}

\section{Finding maximal partitions}
\label{Finding maximal partitions}
In this section, we describe how to use ZDDs and ZDD operations to
find a maximal partition that starts in some gate $g_i$. The ZDDs are defined over the $nm$ variables $vp$ for each $v \in V$ and each $p \in P$.  Each ZDD
represents a family of finite subsets, and each subset $\alpha$
represents a partial mapping $\varphi : V \to P$, where
$\varphi(v) = p$, if and only if $vp \in \alpha$ where $\alpha$ is a mapping.

\begin{example}
  The ZDD $f = \{\{aA, bB, cC, dD\}\}$ represents a family of a single mapping
  that maps $a \mapsto A$, $b \mapsto B$, $c \mapsto C$, and $d \mapsto D$.
\end{example}

First, we define some general sets, which are used throughout the following
operations.  It is sufficient to initialize these sets once at the beginning of
the algorithm.  The set
\begin{equation}
  \label{eq:fromv}
  \FROM(v) = \bigcup_{p \in P} \epsilon_{vp}
\end{equation}
contains all singleton mappings $v \mapsto p$ for some $v \in V$.  Analogously,
the set
\begin{equation}
  \label{eq:tov}
  \TO(p) = \bigcup_{v \in V} \epsilon_{vp}
\end{equation}
contains all singleton mappings $v \mapsto p$ for some $p \in P$.

\begin{example}
  For the example circuit and device we have $\FROM(a) = \{\{aA\}, \{aB\},
  \{aC\}, \{aD\}\}$ and $\TO(A) = \{\{aA\}, \{bA\}, \{cA\}, \{dA\}\}$.
\end{example}

Using this set, we can define two other helpful sets, $\VALID$ and $\BAD$, using
the ZDD family algebra operations.  The set $\VALID$ contains all two-element
partial mappings that are feasible with respect to the coupling constraints of
the device: 
\begin{equation}
  \VALID = \bigcup_{\{p, q\} \in E} \TO(p) \sqcup \TO(q)
\end{equation}

The set $\BAD$ contains all two element sets of illegal partial mappings,
because they either contain an element with two images or two elements which map
to the same image:
\begin{equation}
  \BAD = \bigcup_{v\in V} \binom{\FROM(v)}{2} \cup \bigcup_{p \in P} \binom{\TO(p)}{2}
\end{equation}

\begin{lemma}
  $\beta \in \BAD$ if and only if either $\beta = \{vp, vq\}$ and
  $p \neq q$ or $\beta = \{vp, wp\}$ and $v \neq w$, for all
  $v, w \in V$ and all $p, q \in P$.
\end{lemma}

\begin{proof}
Note that 
\[
\binom{f}{2} = (f \sqcup f) \setminus f
\]
 if $f$ only contains singletons. Therefore, $\binom{\FROM(v)}{2}$ contains all sets $\{vp, vq\}$ where $p \neq q$. Similarly, $\binom{\TO(p)}{2}$ contains all sets $\{vp, wp\}$ for $v \neq w$. For the other direction, note that all pairs are included due to the construction of $\FROM$ and $\TO$ in Eqn.~\ref{eq:fromv} and Eqn.~\ref{eq:tov}.   
  
\end{proof}

\begin{corollary}
  A set $\alpha$ represents a partial mapping if and only if all two-element, illegal mappings are removed where a single pseudo qubit is assigned to multiple physical qubits or where multiple pseudo qubits are assigned to a single physical qubit. Thus, there exists no
  $\beta \in \BAD$ such that $\beta \subseteq \alpha$, respectively.
\end{corollary}

Last, we define the set $\MAP(i)$ which represents all possible
mappings of the pseudo qubits in gate $g_i = \{v, w\}$:
\begin{equation}
  \label{eq:map}
  \MAP(i) = (\FROM(v) \sqcup \FROM(w)) \cap \VALID
\end{equation}

In other words, we first join all possible mappings of $v$ with all
possible mappings of $w$, before we restrict the result two-element
subsets to those which are valid with respect to the target device.

\begin{example}
  Gate $g_1 = \{a, b\}$ can be mapped in eight different ways:
  \[
    \begin{aligned}
      \MAP(1) & = \{\{aA, bB\}, \{aB, bC\}, \{aC, bD\}, \{aD, bA\} \\
      & \phantom{{}=\{{}} \{aB, bA\}, \{aC, bB\}, \{aD, bC\}, \{aA, bD\}\}
    \end{aligned}
  \]
  Also gate $g_2 = \{b, c\}$ can be mapped in eight different ways:
  \[
    \begin{aligned}
      \MAP(2) & = \{\{bA, cB\}, \{bB, cC\}, \{bC, cD\}, \{bD, cA\} \\
      & \phantom{{}=\{{}} \{bB, cA\}, \{bC, cB\}, \{bD, cC\}, \{bA, cD\}\}
    \end{aligned}
  \]
\end{example}

Finally, the possible mappings of two consecutive gates $g_i$ and
$g_{i+1}$ can be computed using
\begin{equation}
  \label{eq:map-consec}
  (\MAP(i) \sqcup \MAP(i+1)) \nonsupersets \BAD.
\end{equation}

\begin{example}
  Recall the two families $\MAP(1)$ and $\MAP(2)$ from the previous
  example.  Joining them leads to a family with up to 64 subsets, out
  of which many do not represent legal partial mappings such as
  $\{aA, bB, bA, cB\}$.  These can be removed using the restriction to
  $\BAD$, resulting a family consisting of the only eight legal
  partial mappings:
  \begin{multline*}
    \{\{aA, bB, cC\}, \{aB, bC, cD\}, \{aC, bD, cA\}, \\ \{aD, bA, cB\},
    \{aA, bD, cC\}, \{aB, bA, cD\}, \\ \{aC, bB, cA\}, \{aD, bC, cB\}\}
  \end{multline*}
\end{example}

\begin{algorithm}[t]
  \KwData{Gate sequence $g_1, \dots, g_k$, and device $(P, E)$}
  \KwResult{partitions $G_j$ with begin and end indexes $b_j, e_j$; all possible mappings $\Phi_j$}
  Set $j \gets 1$, $b_j \gets 1$, $m \gets \MAP(1)$\;
  \For{$i = 2, \dots, k$}{%
    Set $m' \gets (m \sqcup \MAP(i)) \nonsupersets \BAD$\;
    \uIf{$m' = \emptyset$}{%
      \For{$layer$ in $layers$}{%
        calculate scores\;
      }
      \uIf{max score $ \neq 0$}{%
         insert SWAP circuit\;
         update topology\;
         Set $m($max score$)' \gets (m \sqcup \MAP(i)) \nonsupersets \BAD$\;
         Set $m \gets m'$\;

      }\Else{%
        Set $e_j \gets i - 1$, $\Phi_j \gets m$\;
        Set $j \gets j + 1$, $b_j \gets i$, $m \gets \MAP(i)$\;      
      }
    } \Else {%
      Set $m \gets m'$\;
    }
  }
  Set $e_j \gets k$, $\Phi_j \gets m$\;
  \caption{Find maximal partitions}\label{max-partitions}
\end{algorithm}

In some instances, Eqn.~\ref{eq:map-consec} results in the empty set, $\emptyset$, whenever the mappings of two consecutive gates are combined. In this case, SWAP procedures must be performed on physical qubits to exchange pseudo qubit information and extend the partition $G_j$. 

Finding an ideal SWAP circuit during quantum circuit mapping is an intractable problem that works in an exponentially-growing state space. In our method, we narrow this search space by focusing on the implementation of all SWAP circuits that can be executed in parallel during a single time cycle. We call this particular set of SWAP circuits $\LAYERS$, and it is a combination set rather than permutation set of ordered SWAP operations. ZDDs are a good data structure for combinatorial set representation, so a ZDD is implemented to represent the $\LAYERS$ set. Acceptable SWAP operations are determined by the topology of the device, and the sets of operations that can be executed simultaneously within a time cycle are desired. ZDDs are used to create a set of all ``good'' SWAP
circuits, which are those that interact with at least one qubit in the
image of $\varphi$ and the depth of the circuit is one, i.e., a SWAP
circuit may only contain multiple SWAP gates as long as qubits between gates are not shared.  The ZDDs for this task differ from those used to enumerate all mappings, and they do not share any variables. The SWAP circuit ZDDs are defined over the $|E|$ variables $e$ for
each $e \in E$ since SWAP gates can only be placed on certain edges connecting physical qubits according
to the quantum device operational characteristics.  We initialize the ZDD base, or the 1-terminal node, with the following
ZDD.  For each $p \in P$, the ZDD
\begin{equation}
  \label{eq:edges}
  \EDGES(p) = \bigcup\{\epsilon_e \mid \text{$e \in E$ s.t. $p \in e$}\}
\end{equation}
contains all SWAPs that interact with qubit $p$.

All possible subsets of SWAP gates that can be executed in parallel
(i.e., in depth 1) are described by the ZDD
\begin{equation}
  \label{eq:2}
  \LAYERS = \wp \nonsupersets \bigcup_{p\in P} \binom{\EDGES(p)}{2}.
\end{equation}

\begin{example}
The set of SWAP gates that can be parallelized for the topology in Fig.~\ref{fig:circarch}(b) can be determined using Eqn.~\ref{eq:2}. First, the set of edges for all physical qubits $p$ in $P$ must be defined:
\[
edges(A) = \{\{AB\},\{AD\}\},
\]
\[
edges(B) = \{\{AB\},\{BC\}\},
\]
\[
edges(C) = \{\{BC\},\{CD\}\},
\]
\[
edges(D) = \{\{CD\},\{AD\}\}.
\]

\noindent Next, a set must be created that consists of the union of all of the $\binom{N}{k}$ operations where $N = edges(p)$ and $k=2$:
\[
\bigcup_{p\in P} \binom{\EDGES(p)}{2} = \{\{AB,AD\},\{AB,BC\},\{BC,CD\},\{CD,AD\}\}.
\]
\noindent Finally, the nonsupersets operation is implemented between the tautology, or universal family of all subsets for the physical qubit edges. The set $\wp$ is defined as
 \begin{multline*}
\wp = \{ \emptyset, \{AB\},\{BC\},\{CD\},\{AD\},\{AB,BC\},\{AB,CD\},\{AB,AD\},\{BC,CD\},\{BC,AD\}, \\ \{CD,AD\}, \{AB,BC,CD\},\{AB,BC,AD\},\{AB,CD,AD\},\{BC,CD,AD\},\{AB,BC,CD,AD\}\},
  \end{multline*}

\noindent allowing the $ \LAYERS$ for the device in \ref{fig:circarch}(b) to be calculated as

\[
\LAYERS =  \{ \emptyset, \{AB\},\{BC\},\{CD\},\{AD\},\{AB,CD\},\{BC,AD\}\}.
\]
\end{example}

Not all combinations of SWAP gates in $\LAYERS$ may be useful for extending a partition of gates. For example, some SWAP circuits may allow the partition $G_j$ to extend further and have greater depth while others provide more mapping options within $\Phi_j$. As a result, a scoring function is calculated for each member of \textit{layers} to determine the optimal SWAP decision every time it is desired to extend a circuit partition. The function of 

\begin{equation}
\label{eq:score}
score = \left(A\alpha + B\beta\right)\frac{\gamma}{C}
\end{equation}

\noindent is used to select the optimal SWAP circuit according priority weights set by the user and feature counts of the circuit being mapped. In Eqn.~\ref{eq:score}, $\alpha$ is depth weight and $A$ is depth count where depth count in this case describes how many additional two-qubit gates the current partition could cover until the end of the partition is reached if a specific SWAP circuit is implemented. The variable $\beta$ is map weight and $B$ is map count that describes how many maps of pseudo to physical qubits, $\varphi : V \to P$, will be available in $\Phi_j$ if a specific swap circuit is implemented. Finally $\gamma$ is SWAP weight and and $C$ is SWAP operation count of the implemented SWAP circuit. In Eqn.~\ref{eq:score}, SWAP count has an inverse relationship with a layer's \textit{score} as lower overall gate counts, or gate volume, in a technology-mapped implementation are preferred. The parameter $\gamma$, however, can be adjusted to make the cost of an additional SWAP operations less severe. The weights of this function can be tuned to prioritize the growth of the partition with respect to either gate coverage or available mappings if a SWAP circuit layer in \textit{layers} is implemented. In the case that \textit{score} for each layer is zero, the current partition cannot be extended and a new partition must start to continue to map the circuit.

Algorithm~\ref{max-partitions} implements ZDD data structures and the aforementioned ZDD operations to
compute maximal partitions of quantum circuits starting from the first
gate.  A counter $j$ indicates the current partition
number as the algorithm parses through the operators in the network.  In $m$, a set of mappings are stored and updated as the current partition increases in size. These maps are eventually stored in $\Phi_j$.  For
each gate $i$ we try to extend $m$ by adding the gate using $\MAP(i)$,
and storing the resulting mappings in $m'$.  If $m'$ is empty, $layers$ will be used to determine if a SWAP operation can be implemented in order to increase the current partition. The $scores$ for all of the sets in $layers$ are calculated, and the SWAP circuit with the largest $score$ is used to extend the partition. After the SWAP is implemented, the topology of the device is updated, maximum $score$ $m'$ is calculated, and $m$ is updated with $m'$. If the maximum $score$ is zero, then the then the current partition ends at $i-1$, and a new partition $j+1$ starts at gate $i$.  

\section{ZDD mapping in the quantum compilation flow}
\label{ZDD mapping in the Quantum Compilation Flow}

The algorithm discussed in the previous section implements the mapping of pseudo qubits in a quantum circuit specification to physical qubits on a real device. While the mapping procedure is essential for for quantum compilation, additional optimization steps can further improve the technology-mapped logic. For this reason, we propose the incorporation of the ZDD mapping techniques into a larger logic synthesis flow. In this procedure, mapping would occur after a circuit has been decomposed into one- and two-qubit operators and before a specification is compiled by a device's custom compilation tool. Completing synthesis with available compilation tools allows the opportunity to take advantage of existing optimization algorithms while the operators of the circuit are transformed into a platform's native gate library. Additionally, if the maximal partition found by the ZDD mapper does not cover the entire circuit, the native compiler of the technology platform is required to ensure that the placement of the two-qubit operators does not violate the coupling constraints of the device. It should be noted that although the ZDD mapping algorithm was evaluated using superconducting qubits as a target platform, the techniques described here are applicable to other quantum technologies that have coupling restrictions.

\section{Experimental results}

The ZDD quantum mapping algorithm was developed in C++ and was incorporated in the \textit{tweedledum} logic synthesis library found in reference~\cite{Twee}. A subset of benchmarks from~\cite{Feyn} were selected to evaluate the methods described in this work. These benchmarks contain a variety of reversible arithmetic and quantum algorithms such as a Grover's algorithm oracle, a demonstration of the Quantum Fourier Transform (QFT), and various Toffoli implementations. The benchmark set was chosen for experimentation because they are functions commonly seen in quantum synthisis literature and are publicly available for use. These benchmarks are originally specified in a .qc file format with a gate set that contains physically unrealizable multi-qubit gates. Thus, the specifications are transformed into the Clifford+T library of single- and two-qubit gates using the ``phasefold'' pass of the \textit{Feynman} toolkit~\cite{Feyn} . After this procedure, the benchmarks are in a technology-independent form that consists of elementary gates. Mapping to a target quantum device may now begin.

A ring topology was chosen as the target device during synthesis. A ring structure was chosen because this type of architecture is seen in commercial QPUs like in Rigetti's \textit{Agave} and \textit{Aspen} quantum machines. Additionally, the benefit of this structure is that it allows for a homogeneous testing environment that is flexible in size while using benchmarks that vary in number of pseudo qubits. Each benchmark was targeted to a device that contained $n$ physical qubits where $n$  is the number of pseudo qubits in a quantum algorithm. In these devices, all qubits are connected to their two adjacent neighbors, as seen in Fig.~\ref{fig:circarch}(b), and the connections are bidirectional with respect to the placement of the two-qubit $CX$ gate. Once the circuit and topology are selected, Algorithm~\ref{max-partitions} is applied to map the pseudo qubits in the design to physical qubits. The scoring operation of Eqn.~\ref{eq:score} that chooses between the SWAP circuits in $layers$ to extend the partition is set to zero, one, and one for the depth, map, and SWAP weights, respectively. If the benchmark can be covered by an entire partition during the application of Algorithm~\ref{max-partitions}, then the resulting specification is in a fully technology-mapped form and therefore compatible with the available connections of the target device. If multiple partitions are needed for a benchmark, the resulting specification is mapped using a permutation of the largest partition, making additional SWAP operations required for the design to be fully compatible with the target technology. This additional circuit modification is accomplished by the custom compilers that are provided with the Rigetti and IBM SDKs. Compiling the ZDD mapped circuits into the selected device topology with the available Rigetti and IBM compilers is the final procedure in the mapping flow. This step also ensures that all algorithm operations are translated into gates appropriate for the target device. Note that such a translation cannot lead to any further violations of the coupling constraints. After the final compilation step, the circuits are ready for execution on their respective platform since they are in a technology-mapped and optimized form. Details about the benchmarks along with experimental results of the mapping and compilation procedures can be found in Table~\ref{table:results}.

In Table~\ref{table:results}, details about gate depth, gate volume, and two-qubit gate count have been included for the ZDD mapped and compiled circuits. Metrics for the benchmarks after transformation with just the ZDD mapping procedures are also shown for reference. It should be noted that circuits transformed with only the ZDD mapper can include additional SWAP circuitry to expand the mapping partitions, and only if the maximal partition covers the entire benchmark is the resulting circuit fully mapped for the target technology. Because it is often the case that multiple partitions cover a benchmark and the maximum partition must be chosen for pseudo to physical qubit mapping, transformation by the native compiler is required. This synthesis procedure also confirms that the benchmark circuits use the appropriate gate library for the targeted technology. The benchmarks were compiled with and without preprocessing the circuit with the ZDD mapper. Circuits that improved in metrics for a particular device and benchmark whenever ZDD mapping was implemented are emphasized. On average, benchmarks mapped to a Rigetti ring topology saw a decrease of around 10\% with respect to gate depth, gate volume, and two-qubit gate count whenever ZDD mapping was included in technology-dependent logic synthesis flow before compilation. The IBM-compiled circuits, however, only saw an average decrease of 2.3\% in gate depth, gate volume, and two-qubit gate count whenever ZDD mapping was used. Individual improvements in circuit metrics of up to approximately a 50\% decrease was seen in gate volume on the Rigetti devices and up to approximately a 44\% decrease was seen in depth on the IBM devices. These findings demonstrate the potential that ZDD mapping techniques have with respect to finding more optimal solutions whenever generating technology-dependent forms of quantum circuits.

\begin{table*}[!htbp]
\centering
\caption{Gate depth, gate volume, and two-qubit metrics of benchmarks after zdd mapping, IBM compiling, and Rigetti compiling. Values that decreased whenever ZDD mapping was implemented before compilation have been emphasized.}
{\scriptsize
\begin{tabular}{c c r c | c c | c c }
\hline\hline
Benchmark & No. & & Original ZDD& Original Rigetti & Original IBM & ZDD Mapped/ & ZDD Mapped/  \\
Name & Qubits & &Mapped & Compiled &Compiled & Rigetti Compiled & IBM Compiled \\
\hline\hline
barenco\textunderscore tof\textunderscore 3 & 5  &\textit{depth:} & 64  &118  &62 &98\cellcolor{mygrey} \textit{(-16.95\%)} & 84 \textit{(+35.48\%)}  \\
&  & \textit{vol.:}  &95& 446 &180&221\cellcolor{mygrey} \textit{(-50.45\%)} &165 \textit{(-8.33\%)}\cellcolor{mygrey}\\
 &  & \textit{2q gates:}& 73  &68&67& 58\cellcolor{mygrey} \textit{(-14.71\%)}&63 \cellcolor{mygrey}\textit{(-5.97\%)} \\
 \hline
 
barenco\textunderscore tof\textunderscore 4 & 7 &\textit{depth:} &94 &230 & 131 & 155\cellcolor{mygrey} \textit{(-32.6\%)}&130 \textit{(-0.76\%)} \cellcolor{mygrey}\\
&  & \textit{vol:}& 190& 763 & 462 &449 \textit{(-41.15\%)} \cellcolor{mygrey}&335 \textit{(-27.49\%)}\cellcolor{mygrey}\\
 &  &\textit{2q gates:} &152 & 123& 177 &117 \cellcolor{mygrey}\textit{(-4.88\%)}&132 \cellcolor{mygrey}\textit{(-25.42\%)}\\
 \hline
 
barenco\textunderscore tof\textunderscore 5 & 9 &\textit{depth:}& 94 & 231 &121 &155 \textit{(-32.9\%)} \cellcolor{mygrey}&130  \textit{(+7.44\%)}\\
 &  &\textit{vol.:}& 285 &1136 &528&682 \textit{(-39.96\%)} \cellcolor{mygrey}& 505 \cellcolor{mygrey}\textit{(-4.36\%)}\\
&  &\textit{2q gates:}& 231 &184 &201&177 \textit{(-3.8\%)} \cellcolor{mygrey}&201 \textit{(+0\%)} \\
\hline

gf2\textsuperscript{$\wedge$}4\textunderscore mult & 12&\textit{depth:} &46&337 &251 &361 \textit{(+7.12\%)} &354 \textit{(+41.03\%)} \\
 & &\textit{vol:} &232&2319 &1450 & 2593 \textit{(+11.82\%)}&1511 \textit{(+4.21\%)}\\
  & &\textit{2q gates:} &145&363 &557 & 430 \textit{(+18.46\%)}&587 \textit{(+5.39\%)}\\
\hline

gf2\textsuperscript{$\wedge$}5\textunderscore mult & 15&\textit{depth:} &64&422 &259 &504 \textit{(+19.43\%)} &342 \textit{(+32.05\%)}  \\
 & &\textit{vol:} &363&3747 & 2212& 4510 \textit{(+20.36\%)}& 2351 \textit{(+6.28\%)} \\
  & &\textit{2q gates:} &230&596 &842 & 775 \textit{(+30.03 \%)} &910 \textit{(+8.07\%)}\\
\hline

grover\textunderscore 5 & 9&\textit{depth:} &210 &968  &989&872 \textit{(-9.92\%)}\cellcolor{mygrey} &552 \textit{(-44.19\%)}\cellcolor{mygrey} \\
 & &\textit{vol:}& 777 & 4857& 2909&4484 \textit{(-7.68\%)}\cellcolor{mygrey} & 2590 \textit{(-10.97\%)}\cellcolor{mygrey}\\
  & &\textit{2q gates:} &441 & 781& 1096& 739 \textit{(-5.38\%)}\cellcolor{mygrey}&1011 \textit{(-7.76\%)} \cellcolor{mygrey}\\
\hline 

hwb6 & 7&\textit{depth:} &113& 449&269 &432 \textit{(-3.79\%)}\cellcolor{mygrey}&290 \textit{(+7.81 \%)} \\
 & &\textit{vol:} &303&2032 &1049 &2027 \textit{(-0.25\%)}\cellcolor{mygrey}& 1101 \textit{(+4.96\%)}\\
  & &\textit{2q gates:}& 185& 332& 404 &338 \textit{(+1.81\%)}&422 \textit{(+4.46\%)}\\
\hline

mod\textunderscore mult \textunderscore 55 & 9&\textit{depth:}& 49  &189 &123&177 \textit{(-6.35 \%)}\cellcolor{mygrey}&144 \textit{(+17.07 \%)} \\
 & &\textit{vol:} &155 &978 &500&850 \textit{(-13.09 \%)}\cellcolor{mygrey}& 469 \textit{(-6.2\%)}\cellcolor{mygrey}\\
  & &\textit{2q gates:}& 88 &151 &193&143 \textit{(-5.3 \%)}\cellcolor{mygrey}&176 \textit{(-8.81\%)} \cellcolor{mygrey}\\
\hline

mod\textunderscore 5 \textunderscore 4 & 5&\textit{depth:} &60 &115 &94&95 \textit{(-17.4\%)}\cellcolor{mygrey}&92 \textit{(-2.13\%)}\cellcolor{mygrey} \\
 & &\textit{vol:} &121 &459 &229&308 \textit{(-32.9\%)}\cellcolor{mygrey}&239 \textit{(+4.37\%)} \\
  & &\textit{2q gates:}& 98&73 &88&79 \textit{(+8.21\%)}&92 \textit{(+4.55\%)} \\
\hline

qft \textunderscore 4 & 5&\textit{depth:}&142 & 162 & 155& 137 \textit{(-15.43\%)}\cellcolor{mygrey}&105 \textit{(-32.26\%)} \cellcolor{mygrey}\\
 & &\textit{vol:}& 247& 447 & 322&433 \textit{(-3.13\%)}\cellcolor{mygrey}&293 \textit{(-9.01\%)}\cellcolor{mygrey} \\
  & &\textit{2q gates:}&120 & 79 & 126&92 \textit{(+16.46\%)}&114 \textit{(-9.52\%)}\cellcolor{mygrey}\\
\hline

tof\textunderscore 3 & 5  &\textit{depth:} &39 & 98 & 62 &72 \textit{(-26.53\%)}\cellcolor{mygrey}&61 \textit{(-1.61\%)}\cellcolor{mygrey} \\
&  & \textit{vol.:} &75 & 309 &145&195 \textit{(-36.89\%)}\cellcolor{mygrey} &135 \textit{(-6.9\%)}\cellcolor{mygrey} \\
 &  & \textit{2q gates:}& 54 & 47&53&45 \textit{(-4.26\%)}\cellcolor{mygrey}&52 \textit{(-1.89\%)}\cellcolor{mygrey} \\
 \hline

tof \textunderscore 4 & 7&\textit{depth:}& 46 & 117 & 98 &88 \textit{(-24.79\%)}\cellcolor{mygrey}&62 \textit{(-36.73\%)}\cellcolor{mygrey}\\
 & &\textit{vol:} &125 & 505 & 326&327 \textit{(-35.25\%)}\cellcolor{mygrey}&218\textit{(-33.12\%)}\cellcolor{mygrey}\\
  & &\textit{2q gates:}& 92 & 80 & 121 &75 \textit{(-6.25\%)}\cellcolor{mygrey}&84 \textit{(-30.58\%)}\cellcolor{mygrey}\\
\hline

tof \textunderscore 5 & 9&\textit{depth:} &46 & 118& 68&89 \textit{(-24.58\%)}\cellcolor{mygrey}&62 \textit{(-8.82\%)}\cellcolor{mygrey}\\
 & &\textit{vol:}& 175 &707 &335&459 \textit{(-35.08\%)}\cellcolor{mygrey}&308 \textit{(-8.06\%)}\cellcolor{mygrey} \\
  & &\textit{2q gates:}& 130 & 112 &132 &106 \textit{(-5.36\%)}\cellcolor{mygrey}&118 \textit{(-10.61\%)}\cellcolor{mygrey}\\
\hline

vbe \textunderscore adder \textunderscore 3 & 10&\textit{depth:}& 67& 216&165&197 \textit{(-8.8\%)}\cellcolor{mygrey}& 232 \textit{(+40.61\%)}\\
 & &\textit{vol:} &162 &1244 & 765&1131 \textit{(-9.08\%)}\cellcolor{mygrey}&835 \textit{(+9.15\%)} \\
  & &\textit{2q gates:} &122& 190& 294&195 \textit{(+2.63\%)}&329 \textit{(+11.9\%)} \\
\hline\hline
\end{tabular}}
\label{table:results}
\end{table*}

\section{Conclusion}
\label{Conclusion}

We present a method for mapping quantum logic circuits to actual devices using ZDDs. The required operations and algorithm are described, and the implementation is developed and tested in a quantum compilation flow that targets devices meant for the Rigetti and IBM families of superconducting quantum computers. When experimental results are evaluated, it is observed that incorporating ZDD mapping into quantum logic synthesis in many cases allowed for more optimal circuits to be found in their technology-dependent form. These results suggest that the ZDD may be a useful tool for quantum compilation that should be continued to be investigated in the future.

\paragraph*{Acknowledgements.} The author thanks Alan Mishchenko for helpful discussion and comments. This work was supported by the Swiss National Science Foundation in the project "Technology-dependent Optimization in Quantum Compilation" (IZSEZ0_184016).

\appendix

\section{Implementation of ZDD operations}

In~\cite[Ex.~7.1.4-207]{Knuth4A}, Knuth describes the implementation
of a ZDD operation $f \mathbin{\S} k$, which is similar to the
operation $\binom{f}{k}$, which is used in this paper.  No description
of an implementation for $\binom{f}{k}$ was found in the literature,
and therefore we report our implementation here, in a similar style
Knuth used to describe the implementation of $f \mathbin{\S} k$.

$\CHOOSE(f, k)$ = ``If $k = 1$, return $f$.  If $f = \emptyset$ and
$k>0$, return $\emptyset$. If $f = \emptyset$ and $k=0$, return
$\{\emptyset\}$.  If $\binom{f}{k} = r$ is in the cache, return $r$.
Otherwise set $r \gets \CHOOSE(f_l, k)$.  If $k > 0$, set
$q \gets \CHOOSE(f_l, k - 1)$ and $r \gets \ZUNIQUE(f_v, r, q)$. Put
$\binom{f}{k} = r$ in the cache, and return $r$.''

\bibliographystyle{eptcs}
\bibliography{library}

\end{document}